\documentclass[11pt]{article}
\usepackage{verbatim,amsmath,amsthm,eufrak}
\usepackage{cite} % prevents showidx from correct work
\usepackage{mathrsfs}
\usepackage{makeidx}
\usepackage{yhmath}
\usepackage{enumitem}
\usepackage{color,xcolor}
\usepackage[hypertexnames=false,hyperfootnotes=false,colorlinks=true,linkcolor=blue,%
citecolor=purple,filecolor=magenta,urlcolor=cyan,unicode,linktocpage=true,pagebackref=false]{hyperref}
\usepackage{scalerel,stackengine}

\renewcommand{\arraystretch}{1.2}
\makeatletter
\newdimen\normalarrayskip              % skip between lines
\newdimen\minarrayskip                 % minimal skip between lines
\normalarrayskip\baselineskip
\minarrayskip\jot
\newif\ifold             \oldtrue            \def\new{\oldfalse}
\def\arraymode{\ifold\relax\else\displaystyle\fi} % mode of array entries
\def\eqnumphantom{\phantom{(\theequation)}}     % right phantom in eqnarray
\def\@arrayskip{\ifold\baselineskip\z@\lineskip\z@
     \else
     \baselineskip\minarrayskip\lineskip2\minarrayskip\fi}
\def\@arrayclassz{\ifcase \@lastchclass \@acolampacol \or
\@ampacol \or \or \or \@addamp \or
   \@acolampacol \or \@firstampfalse \@acol \fi
\edef\@preamble{\@preamble
  \ifcase \@chnum
     \hfil$\relax\arraymode\@sharp$\hfil
     \or $\relax\arraymode\@sharp$\hfil
     \or \hfil$\relax\arraymode\@sharp$\fi}}
\def\@array[#1]#2{\setbox\@arstrutbox=\hbox{\vrule
     height\arraystretch \ht\strutbox
     depth\arraystretch \dp\strutbox
     width\z@}\@mkpream{#2}\edef\@preamble{\halign
\noexpand\@halignto
\bgroup \tabskip\z@ \@arstrut \@preamble \tabskip\z@ \cr}%
\let\@startpbox\@@startpbox \let\@endpbox\@@endpbox
  \if #1t\vtop \else \if#1b\vbox \else \vcenter \fi\fi
  \bgroup \let\par\relax
  \let\@sharp##\let\protect\relax
  \@arrayskip\@preamble}
\def\eqnarray{\stepcounter{equation}%
              \let\@currentlabel=\theequation
              \global\@eqnswtrue
              \global\@eqcnt\z@
              \tabskip\@centering
              \let\\=\@eqncr
 \halign to \displaywidth\bgroup
    \eqnumphantom\@eqnsel\hskip\@centering
    $\displaystyle \tabskip\z@ {##}$%
    \global\@eqcnt\@ne \hskip 2\arraycolsep
         $\displaystyle\arraymode{##}$\hfil
    \global\@eqcnt\tw@ \hskip 2\arraycolsep
         $\displaystyle\tabskip\z@{##}$\hfil
         \tabskip\@centering
    &{##}\tabskip\z@\cr}
\begingroup\ifx\undefined\newsymbol \else\def\input#1 {\endgroup}\fi

\newcounter{app}

\def\app{\setcounter{equation}{0}
\def\theequation{A\Roman{app}.\arabic{equation}}\par
   \addvspace{4ex}
   \@afterindentfalse
  \secdef\@app\@dapp}
\newcommand\@app{\@startsection {app}{1}{0ex}%
                                   {-3.5ex \@plus -1ex \@minus -.2ex}%
                                   {2.3ex \@plus.2ex}%
                                   {\normalfont\Large\bf}}

\def\@dapp#1{%
{\parindent \z@ \raggedright  \bf #1}\par\nobreak}
\def\l@app#1#2{\ifnum \c@tocdepth >\z@
    \addpenalty\@secpenalty
    \addvspace{1.0em \@plus\p@}%
    \setlength\@tempdima{8.5em}%
    \begingroup
      \parindent \z@ \rightskip \@pnumwidth
      \parfillskip -\@pnumwidth
      \leavevmode \bfseries
      \advance\leftskip\@tempdima
      \hskip -\leftskip
      #1\nobreak\hfil \nobreak\hb@xt@\@pnumwidth{\hss #2}\par
    \endgroup\fi}
\newcounter{sapp}[app]

\def\sapp{\def\theequation{A\arabic{app}.\arabic{equation}}\par
   \@afterindentfalse
  \secdef\@sapp\@dsapp}
\newcommand\@sapp{\@startsection{sapp}{2}{\z@}%
                                     {-3.25ex\@plus -1ex \@minus -.2ex}%
                                     {1.5ex \@plus .2ex}%
                                     {\normalfont\large\bfseries}}

\def\@dsapp#1{%
{\parindent \z@ \raggedright  \bf #1}\par\nobreak}
\newcommand{\l@sapp}{\@dottedtocline{2}{1.5em}{3em}}
\def\draft{\oddsidemargin -.5truein
        \def\@oddfoot{\sl preliminary draft \hfil
        \rm\thepage\hfil\sl\today\quad\militarytime}
        \let\@evenfoot\@oddfoot \overfullrule 3pt
        \let\label=\draftlabel
        \let\marginnote=\draftmarginnote
   \def\@eqnnum{(\theequation)\rlap{\kern\marginparsep\tt\@eqnlabel}%
\global\let\@eqnlabel\@vacuum}  }

\def\be{\begin{eqnarray}}
\def\ee{\end{eqnarray}}
\def\nn{\nonumber}

\def\p{\partial}

\def\beq{\begin{equation}}
\def\eeq{\end{equation}}
\def\ba{\beq\new\begin{array}{c}}
\def\ea{\end{array}\eeq}
\def\be{\ba}
\def\ee{\ea}

\def\sspan{{\rm span}\,}

\newfont{\Bbbb}{msbm7 scaled 1\@ptsize00}

\newcommand{\z}{\raise-1pt\hbox{$\mbox{\Bbbb Z}$}}
\newcommand{\n}{\raise-1pt\hbox{$\mbox{\Bbbb N}$}}
\newcommand{\rr}{\raise-1pt\hbox{$\mbox{\Bbbb R}$}}
\newcommand{\cc}{\raise-1pt\hbox{$\mbox{\Bbbb C}$}}
\def\normordboson{ {\scriptstyle {{*}\atop{*}}} }

\def\ch{{\rm ch}}
\newcommand{\<}{\left <}
\renewcommand{\>}{\right >}

\newfont{\alef}{msbm10 at 11pt}
\newfont {\goth}{eufm10 at 11pt}
\def\mathbb#1{\hbox{{\alef #1}}}

\DeclareMathOperator{\GL}{GL}

\let\@@savethanks\thanks
\def\thanks#1{\gdef\thefootnote{\alph{footnote}}\@@savethanks{#1}}

\newtheorem{theorem}{Theorem}%[section]
\newtheorem{lemma}{Lemma}[section]
\newtheorem{proposition}[lemma]{Proposition}
\newtheorem{corollary}[lemma]{Corollary}
\newtheorem{remark}{Remark}[section]

\newtheorem{conjecture}{Conjecture}[section]
\newtheorem*{theorem*}{Theorem}

\numberwithin{equation}{section}

\makeatletter
\g@addto@macro \normalsize {%
 \setlength\abovedisplayskip{14pt plus 3pt minus 3pt}%
 \setlength\belowdisplayskip{14pt plus 3pt minus 3pt}%
  \setlength\abovedisplayshortskip{11pt plus 3pt minus 3pt}%
 \setlength\belowdisplayshortskip{11pt plus 3pt minus 3pt}%
}
\makeatother

\bigskip
\bigskip

\title{
\bigskip
{\bf KP integrability of triple Hodge integrals. I. \\
 From Givental group to hierarchy symmetries} \vspace{.5cm}}
\author{{\bf Alexander Alexandrov}\thanks{E-mail:  {\tt alexandrovsash at gmail.com}}
\date{ } \\
{\small {\it Center for Geometry and Physics, Institute for Basic Science (IBS),
 Pohang 37673, Korea  }}
}

\begin{document}

\setcounter{footnote}{0}

\setcounter{tocdepth}{3}

\maketitle

\vspace{-8.0cm}

\begin{center}
%\hfill ITEP/TH-18/16
\end{center}

\vspace{6.5cm}
\begin{small}\begin{center}
\today
\end{center}\end{small}
%\bigskip
\begin{abstract} 
In this paper, we investigate a relation between the Givental group of rank one and the Heisenberg-Virasoro symmetry group  of the KP hierarchy. We prove, that only a two-parameter family of the Givental operators can be identified with elements of the Heisenberg-Virasoro symmetry group. This family describes triple Hodge integrals satisfying the Calabi-Yau condition.
Using the identification of the elements of two groups we prove that the generating function of triple Hodge integrals satisfying the Calabi-Yau condition and its $\Theta$-version are tau-functions of the KP hierarchy. This generalizes the result of Kazarian on KP integrability in the case of linear Hodge integrals.
\end{abstract}
\bigskip

%{Keywords: enumerative geometry, matrix models, tau-functions, KP hierarchy, Virasoro constraints, cut-and-join operator}\\

\bigskip

{\small \bf MSC 2020 Primary: 37K10, 14N35, 81R10, 14N10; Secondary:  81T32.}

\begin{comment}
(	
37K10 Completely integrable infinite-dimensional Hamiltonian and Lagrangian systems, integration methods, integrability tests, integrable hierarchies (KdV, KP, Toda, etc.)
14N35 Gromov-Witten invariants, quantum cohomology, Gopakumar-Vafa invariants, Donaldson-Thomas invariants (algebro-geometric aspects) 
81R10 Infinite-dimensional groups and algebras motivated by physics, including Virasoro, Kac-Moody, W-algebras
and other current algebras and their representations 
81R12 Groups and algebras in quantum theory and relations with integrable systems
14H70 Relationships between algebraic curves and integrable systems
81T32 Matrix models and tensor models for quantum field theory
05A15 Exact enumeration problems, generating functions
14N10 Enumerative problems (combinatorial problems) in algebraic geometry
15B52 Random matrices (algebraic aspects) 
81T40 Two-dimensional field theories, conformal field theories, etc. in quantum mechanics
\end{comment}

%\bigskip

%\bigskip

\newpage

\tableofcontents

\def\thefootnote{\arabic{footnote}}
\section{Introduction}
%\addcontentsline{toc}{section}{Introduction}
%\def\theequation{\arabic{equation}}
\setcounter{equation}{0}

Since the works of  't Hooft \cite{tHooft} and Br\'ezin, Itzykson, Parisi, and Zuber \cite{BIPZ}, it is known that the matrix models provide an efficient tool for the investigation of enumerative geometry problems. Recent remarkable progress in enumerative geometry/matrix model theory, in particular, is related to the Chekhov-Eynard-Orantin topological recursion \cite{TR2,TR3} and Givental decomposition \cite{Giv1,Giv2}. Both describe a non-trivial extension of the matrix model theory, which ranges far beyond the scope of the classical matrix integrals. Both topological recursion and Givental decomposition are naturally described in terms of summation over graphs, thus formally can be represented by infinite-dimensional Gaussian matrix integrals with integrands made of the Kontsevich-Witten tau-function and its generalizations. 

While a relationship between these two methods is known \cite{DBOSS,CN}, their relation to other classical ingredients of the  matrix model theory is not always clear. Usually, matrix models satisfy two different families of differential (or differential-difference) equations. One family, the { \it linear } equations, is given by the Virasoro or, more generally, W-constraints. For known examples, recursive solution of these equations naturally leads to topological recursion and Givental-type decomposition, see, e.g., \cite{AMM2,AMM4,AMM5,AMM6,AMM8,EynardMM}. Another family of equations is given by the {\it bilinear} Hirota equations of the KP/Toda type integrable hierarchies. While the Virasoro constraints constitute an inherent part of the Givental construction, the situation with integrability is much less clear. Many examples of the enumerative geometry generating functions, which can be described by topological recursion/Givental decomposition (TR/GD), are certain tau-functions of integrable hierarchies, however, a general relation between TR/GD and integrability is not known (see, however, \cite{GM,FGM,Feigin,EyInt}).

The main goal of this paper is to clarify the role of KP/Toda type integrability in the general scheme of TR/GD. It is important, in particular, because some elements of the integrable systems can be naturally identified with the main components of the TR/GD. For example, the quantum spectral curve, playing a central role in TR/GD construction, can be identified with a specific Kac-Schwarz operator, while the wave function, annihilated by this operator, is given by the Baker-Akhiezer function\cite{Aeg,KS1,KS2,KS3,KS4,H3_2}.

In this paper, we address a question of the relation between the Givental group and the symmetry groups of the integrable hierarchies. These are different groups, acting on different spaces. However, sometimes it is possible to identify the elements of two groups up to a linear change of variables. An example of such identification was constructed in the work of Liu and Wang \cite{LiuWang}. They showed that the Givental operator associated with the generating function of linear Hodge integrals is equal to an element of the Heisenberg-Virasoro subgroup of $GL(\infty)$ up to a linear change of variables. 

Therefore, a question arises: are there any other  elements of the Givental group, which can be identified with the elements of the symmetry group of the integrable hierarchy? In this paper, we give a complete answer to this question for the rank one case of the Givental construction on one side and the KP integrable hierarchy on another. By {\em rank} here we mean the rank of the Frobenius manifold = the number of the KdV tau-functions in the Givental formula. The first important result of this paper, Theorem \ref{th_RL}, states that the infinite-dimensional Givental group contains only a two-dimensional family of such operators. 

On the enumerative geometry side, this two-dimensional family is known to describe a very interesting generating function, namely, a generating function of triple Hodge integrals satisfying the Calabi-Yau condition. Hence, from the identification of the Givental and Heisenberg-Virasoro group elements it follows, that after a certain linear change of variables the generating function of triple Hodge integrals satisfying the Calabi-Yau condition becomes a tau-function of the KP hierarchy. Theorem \ref{Th_H}, which describes this KP tau-function, is the second important result of this paper. It generalizes the result of Kazarian for linear Hodge integrals \cite{Kaza}.  Theorem \ref{Norbt} extends this KP integrability to the case with included Norbery's $\Theta$-classes.

We restrict ourselfes to the simplest rank one case. However, we expect that the results of this paper can be immediately generalized to the higher rank Givental groups. This should lead to new interesting examples of integrable generating functions of the cohomological field theory. In the companion paper \cite{H3_2} we investigate a family of the tau-functions of the KP hierarchy, that generalizes the two-parametric family constructed in this paper and conjecturally describes interesting enumerative geometry invariants in the $r$-spin case. This family can be described by a deformation of the generalized Kontsevich model. The current project, in particular, was motivated by \cite{MMKH}, where a question about KP integrability of non-linear Hodge integrals and their description in terms of generalized Kontsevich model was raised.

It is known that enumerative geometry generating functions constitute only a small class of tau-functions of integrable hierarchies. We hope that the identification between the elements of the Givental group and symmetry groups of integrable hierarchies will help us to classify all such tau-functions.

The present paper is organized as follows. In Section \ref{Section2} we consider rank one Givental operators, which after a linear change of variables can be identified with the elements of the Heisenberg-Virasoro symmetry group of the KP hierarchy. 
In Section \ref{S_Hodge} we identify this family with the generating functions of triple Hodge integrals, which proves the KP integrability of the latter.

%%%%%%%%%%%%%%%%%%%%%%%%%%%%%%%%%%%%%%%%%

\section{Two groups of symmetries}\label{Section2}

In this section we consider two different infinite-dimensional groups acting on the bosonic Fock spaces. First of them, the Givental group of the quantized symplectic transformations is a group of symmetries of semi-simple cohomological field theory. Second is a Heisenberg-Virasoro subgroup of the $GL(\infty)$ symmetry group of the KP hierarchy. We restrict the action of both groups to the space of functions of odd times. We prove, that on this space only a two-parametric family of the Givental operators can be identified, up to a linear change of variables, with the elements of the Heisenberg-Virasoro group. 

\subsection{Givental group}\label{Givsec}

Let us consider the Givental quantization scheme for the rank one case. We basically follow the presentation of \cite{Giv2,Lee,Dubrovin16}. However, we consider a different normalization of the variables $T_k$, with an additional factor of $\hbar$. This change of normalization corresponds to the transition from the genus expansion of the generating function
\be
\exp\left(\sum_{g=0}^\infty \sum_{n=1}^\infty \hbar^{2g-2}{\mathcal F}_{g,n}\right),
\ee
where $g$ is the genus, $n$ is the number of the marked points, to the topological expansion 
\be
\exp\left(\sum_{g=0}^\infty \sum_{n=1}^\infty \hbar^{2g-2+n}{\mathcal F}_{g,n}\right),
\ee
because the latter is more convenient for the KP hierarchy description.  

Let ${ H}={\mathbb C}[\![z,z^{-1}]\!]$ be the space of formal Laurent series in an indeterminate $z$. We introduce the symplectic form on this space
\be
\Omega(f,g)= \frac{1}{2\pi i} \oint f(-z)g(z) dz.
\ee
For the natural decomposition 
\be
{ H}={H}_+\oplus{H}_-,
\ee
where ${ H}_+=\sspan_{\cc}\{1,z,z^2,\dots\}$ and ${H}_-=\sspan_{\cc}\{z^{-1},z^{-2},z^{-3},\dots\}$, we introduce the Darboux coordinate system $\{p_k,q_m\}$ with
\be
{\mathcal J}(z):=\sum_{k=0}^\infty \left(q_k z^{k}+(-z)^{-k-1}p_k\right).
\ee
After Givental, for any infinitesimal symplectic transformation $A$ on ${H}$ (that is, a transformation such that $\Omega(Af,g)+\Omega(f,Ag)=0$) we consider
\be
H_A:=\frac{1}{2}\Omega(A{\mathcal J}, {\mathcal J}).
\ee
This defines a Lie algebra isomorphism:
\be
H_{\left[A,B\right]}=\left\{H_A,H_B\right\},
\ee
where the Poisson bracket is given by
\be
\left\{H_A,H_B\right\}:=\sum_{i=0}^\infty\left(\frac{\p H_A}{\p p_i}\frac{\p H_B}{\p q_i}-\frac{\p H_B}{\p p_i}\frac{\p H_A}{\p q_i}\right).
\ee

Using the standard Weyl quantization we quantize these operators to order $\leq 2$ linear differential operators 
\be\label{Gquant}
\widehat{A}=\widehat{H}_A:=\frac{1}{2}\normordboson\Omega(A \widehat{\mathcal J}, \widehat{\mathcal J})\normordboson,
\ee
where
\be\label{Givcur}
\widehat{\mathcal J}(z)=\sum_{k=0}^{\infty}\left( \tilde{T}_kz^k+ (-z)^{-k-1}\frac{\p}{\p T_k}\right),
\ee
and $\normordboson\dots \normordboson$ denotes the standard bosonic normal ordering, which puts all $\frac{\p}{\p T_k}$ to the right of all $T_k$. The so-called {\em dilaton shift}  is given by $\tilde{T}_k=T_k-\hbar^{-1} \delta_{k,1}$.
This gives us a central extension of the original algebra, with the commutator
\be
\left[\widehat{A},\widehat{B}\right]=\widehat{\left[A,B\right]}+\mathcal{C}(H_A,H_B),
\ee
where the so-called 2-cocycle term satisfies
\be
\mathcal{C}\left(p_ip_j,q_k q_m\right)=-\mathcal{C}\left(q_k q_m,p_ip_j\right)=\delta_{i,k}\delta_{j,m}+\delta_{i,m}\delta_{j,k}
\ee
and vanishes for all other elements $H_A$. For the finite symplectic transformations we define
\be
\widehat{e^{A}}:=e^{\widehat{A}}.
\ee

Below we consider only the symplectic transformations of the form $R(z)=1+R_1z+R_2z^2+\dots\in 1+ z{\mathbb C}[\![z]\!]$.  They satisfy the symplectic condition
\be
R(z)R(-z)=1
\ee 
if and only if $\log R(z)$ is odd function of $z$. Let us denote
\be
\widehat{W}_k:=\widehat{z^{2k-1}}, \,\,\,\,k \in  {\mathbb Z},
\ee
then from (\ref{Gquant}) we have \cite{Giv2,Lee} 
\be\label{Woper}
\widehat{W}_k=-\sum_{m}\tilde{T}_m\frac{\p}{\p T_{m+2k-1}}+\frac{1}{2}\sum_{m=0}^{-2k}(-1)^{l+1}\tilde{T}_m\tilde{T}_{-2k-m}+\frac{1}{2}\sum_{m=0}^{2k-2}(-1)^m\frac{\p^2}{\p T_m \p T_{2k-m-2}}.
\ee
These operators satisfy the commutation relations
\be
\left[\widehat{W}_k,\widehat{W}_m\right]=-\frac{2k-1}{2}\delta_{k+m,1}.
\ee
After quantization, using (\ref{Woper}) we obtain 
\be\label{UTS}
\widehat{R}:=\exp\left(\widehat{\log R(z)}\right).
\ee
Below we call the group of such operators the {\em Givental group}.

The Givental group acts of the space of cohomological field theories (CohFT) with flat unit. If we relax the flat unit condition, we can include translations of the 
times, see an example in Section \ref{S_theta}. However, translations are natural symmetries of the KP hierarchy, so this enrichment is trivial from the point of view of the identification of operators of two groups.

Operators $\widehat{R}$ can be factorized, namely, one can factor out the part, corresponding to the linear change of variables and the translation of variables. This factorization, except for simple extraction of the translation operator, is given by Proposition 7.3 in \cite{Giv2}. Here we briefly remind the reader this relation. Let us introduce a formal series in two variables
\be\label{Vgiv}
V^R(z,w):=\frac{1-R(-w)R(-z)}{w+z}.
\ee
The matrix of its coefficients 
\be\label{gV}
V^R(z,w)=\sum_{k,l=0}^\infty V^R_{kl}w^k z^l
\ee
is symmetric, $V_{kl}^R=V_{lk}^R$, and its entries are polynomials in $R_k$. Let the linear change of the shifted variables from ${\bf T}$ to ${\bf T^R}$ be given by
\be\label{trans1}
\sum_{k=0}^\infty \tilde{T}^R_k z^k:=R(-z)\sum_{k=0}^\infty \tilde{T}_kz^k.
\ee
This transformation generates the change of the dilaton shift
\be\label{shiftsR}
\sum_{k=2}^\infty \delta_k z^k:=z(1-R(-z)).
\ee
 Then, for any element of the  Givental group (\ref{UTS})
and any series $Z({\bf T})$ we have \cite{Giv2}
\begin{lemma}[Givental]\label{lem_G}
\be\label{Givf}
\widehat{R}\cdot Z({\bf T})=\,e^{\frac{1}{2}\sum_{i,j=0}^\infty{V_{ij}^R}\frac{\p^2}{\p T_i \p T_j}}\,\left.e^{\hbar^{-1} \sum_{k=2}^\infty \delta_k \frac{\p}{\p T_k}}\cdot Z({\bf T})\right|_{{T}_k\mapsto {T}^R_k}.
\ee
\end{lemma}
In this expression one acts by the operator on the function $Z({\bf T})$, and then in the result substitutes ${T}_k$ with ${T}^R_k$.

%%%%%%%%%%%%%%%%%%%%%%%%%%%%%%%%%%%%%%%%%%%%%%%%%%%%%%%%%

\subsection{Heisenberg-Virasoro subgroup of $GL(\infty)$}\label{KPsec}

In this section we briefly describe the Heisenberg-Virasoro subgroup of $\GL(\infty)$ symmetry group of KP hierarchy, for more details see the companion paper \cite{H3_2} and references therein. Let
\be
\widehat{ J}(z):=\sum_{k\in \z} \frac{\widehat{J}_k}{z^{k+1}},
\ee
where
\be
\widehat{J}_k =
\begin{cases}
\displaystyle{\frac{\p}{\p t_k} \,\,\,\,\,\,\,\,\,\,\,\, \mathrm{for} \quad k>0},\\[2pt]
\displaystyle{0}\,\,\,\,\,\,\,\,\,\,\,\,\,\,\,\,\,\,\, \mathrm{for} \quad k=0,\\[2pt]
\displaystyle{-kt_{-k} \,\,\,\,\,\mathrm{for} \quad k<0.}
\end{cases}
\ee
We consider the generating function of  the Virasoro operators
\be
\normordboson \widehat{J}(z)^2\normordboson= 2 \sum_{k\in \z} 
\frac{ \widehat{L}_k}{z^{k+2}}
\ee 
or
\be\label{Vird}
\widehat{L}_m=\frac{1}{2} \sum_{a+b=-m}a b t_a t_b+ \sum_{k=1}^\infty k t_k \frac{\p}{\p t_{k+m}}+\frac{1}{2} \sum_{a+b=m} \frac{\p^2}{\p t_a \p t_b},
\ee
where, depending on the sign of $m$, only the first or the last summation appears. These formulas are the KP analogs of (\ref{Givcur}) and (\ref{Gquant}). The Heisenberg-Virasoro group ${\mathcal V}$ is generated by the operators $\widehat{J}_k$,  $\widehat{L}_k$ and a unit. Let us consider a Virasoro subgroup of ${\mathcal V}$ given by the operators of the form
\be\label{Goper}
\widehat{V}=\exp\left({\sum_{\z_{>0}} a_k \widehat{L}_k}\right) \in {\mathcal V}.
\ee
With any such operator one can associate the operator $e^{{\sum_{\z_{>0}} a_k{\mathtt l}_k}}$,
where
\be
{\mathtt l}_m=-z^m\left(z\frac{\p}{\p z}+\frac{m+1}{2}\right)
\ee
are the generators of the Witt subalgebra of the algebra of of diffeomorphisms on the circle.

For any set $a_k$, $k\in {\mathbb Z}_{>0}$, consider the series
\be\label{ffunct}
f(z):=e^{{\sum_{\z_{>0}} a_k{\mathtt l}_k}}\,  z \,  e^{-{\sum_{\z_{>0}} a_k{\mathtt l}_k}} \in z+z{\mathbb C}[\![z]\!].
\ee
Then
\be
h(z):=e^{-{\sum_{\z_{>0}} a_k{\mathtt l}_k}}\,  z \,  e^{{\sum_{\z_{>0}} a_k{\mathtt l}_k}} \in z+z{\mathbb C}[\![z]\!]
\ee
is an inverse series, $f(h(z))=h(f(z))=z$. There is a one-to-one correspondence between the space of series (\ref{ffunct}) and the subgroup (\ref{Goper}) of the Virasoro group. From the commutation relations of the Heisenberg-Virasoro algebra, 
\begin{align}
\left[\widehat{J}_k,\widehat{J}_m\right]&=k \delta_{k,-m},\nn\\
\left[\widehat{L}_k,\widehat{J}_m\right]&=-m \widehat{J}_{k+m},\\
\left[\widehat{L}_k,\widehat{L}_m\right]&=(k-m)\widehat{L}_{k+m}+\frac{1}{12}\delta_{k,-m}(k^3-k)\nn
\end{align}
 it follows that
\be\label{Gconj}
\widehat{V}\,\widehat{ {J}}(z)\, \widehat{V}^{-1}=h'(z)\widehat{ {J}}(h(z)).
\ee
Let $\widehat{V}_0$ be a part of the element of the upper-triangular subgroup (\ref{Goper}) that describes the linear change of variables,
\be\label{V0}
\widehat{V}_0:=\exp\left(\sum_{k=1}^\infty a_k \sum_{m=1}^\infty m t_m \frac{\p}{\p t_{k+m}}\right).
\ee
With the help of the Campbell-Baker-Hausdorff formula operator $\widehat{V}$ can  be factorized 
\be\label{Virex}
\widehat{V}=\widehat{V}_0\, \exp\left(\frac{1}{2}\sum_{k,m=1}^\infty v_{km}\frac{\p^2}{\p t_k \p t_m}\right)
\ee
for some $v_{ij}\in {\mathbb C}[\![a_1,a_2,a_3,\dots]\!]$. Below we will describe the space of all $v_{ij}$ which can be obtained in this way.

Using the commutation relations of the Virasoro algebra we can always factor out a term $\exp(a_1\widehat{L}_1)$,
\be
\exp\left({\sum_{\z_{>0}} a_k \widehat{L}_k}\right)=\exp(a_1\widehat{L}_1)\,\exp\left({\sum_{\z_{>1}} \tilde{a}_k \widehat{L}_k}\right)
\ee
for some new $\tilde{a}_k$. Operator ${\widehat L}_1=\sum_{k=1}^\infty k t_k \frac{\p}{\p t_{k+1}}$ is a first order differential operator, so $\exp(a_1\widehat{L}_1)$ gives only a linear change of variables. As we are going to identify the Givental operators with the elements of the Heisenberg-Virasoro group up to the linear change of variables, there is a certain arbitrariness in the choice of $a_1$.  We call the transformation, given by $\exp(a_1\widehat{L}_1)$, the {\em gauge transformation}, and the choice of $a_1$ will be called the choice of {\em gauge}. 
For the series (\ref{ffunct}) the gauge transformation is given by
\be\label{gauge}
f(z) \mapsto f(z/(1+a_1z)).
\ee

%%%%%%%%%%%%%%%%%%%%%%%%%%%%%%%%%%%%%%%%%%%

\subsection{Quadratic part of the Virasoro group operators}

To find the matrix $v_{km}$ in (\ref{Virex}) it is enough to act by both sides of (\ref{Virex}) on the function
\be
E:=\exp\left(\sum_{k=1}^\infty k  t_k q_k\right),
\ee
where ${\bf q}$ are auxiliary variables, independent on ${\bf t}$. Let us show that the coefficients $v_{nm}$ in (\ref{Virex}) are the so-called Grunsky coefficients of function $h(z)$.
Consider
\be
v(\eta_1,\eta_2):=\sum_{k,m=1}^\infty v_{km} \eta_1^k \eta_2^m.
\ee
\begin{lemma}\label{lem_log}
\be\label{vv}
v(\eta_1,\eta_2)=\log\left(\frac{h(\eta_1)-h(\eta_2)}{\eta_1-\eta_2}\right).
\ee
\end{lemma}
\begin{remark}
This formula was known to the experts a long time ago, and can be proven via the Campbell-Baker-Hausdorff formula. Here we provide its proof for completeness.
\end{remark}
\begin{proof}
On one hand, the action of the right hand side of (\ref{Virex}) on $E$ yields
\be\label{actiononB}
\widehat{V}\cdot E =\widehat{V}_0 \cdot \exp\left(\sum_{k=1}^\infty k t_k q_k+\frac{1}{2}\sum_{k,m=1}^\infty k m\, v_{km}q_k q_m \right).
\ee
On the other hand 
\begin{equation}
\begin{split}
\widehat{V} \cdot E&=\exp\left(\sum_{k=1}^\infty k \left(\widehat{V}   t_k  \widehat{V}^{-1}\right) q_k\right)\cdot 1\\
&=\exp\left( \sum_{k=1}^\infty \left(k  t_k \tilde{q}_k + \tilde{q}_{-k}\frac{\p}{\p t_k}\right)\right)\cdot 1,
\end{split}
\end{equation}
where from (\ref{Gconj}) we have
\be
\tilde{q}_k:=\frac{1}{2\pi i }\oint_{\infty}\left(\sum_{m=1}^\infty \frac{q_m}{f(z)^m}\right) \frac{dz}{z^{k+1}}, \,\,\,\,\,\,\,\, k \in {\mathbb Z}. 
\ee
Hence
\begin{equation}
\begin{split}
\widehat{V}\cdot E&=\exp\left( \sum_{k=1}^\infty \left(k  t_k \tilde{q}_k +\frac{1}{2} k \tilde{q}_k\tilde{q}_{-k}\right)\right)\\
&=\widehat{V}_0 \cdot \exp\left(\sum_{k=1}^\infty \left( k t_k q_k+\frac{1}{2} k \tilde{q}_k\tilde{q}_{-k}\right)\right).
\end{split}
\end{equation}
Comparing this expression with (\ref{actiononB}) we conclude 
\be
\sum_{k,m=1}^\infty km\,\,  v_{km}q_kq_m= \sum_{k=1}^\infty k \tilde{q}_k\tilde{q}_{-k}.
\ee
Therefore,
\begin{equation}
\begin{split}
v_{km}&=\frac{1}{km}\sum_{j=1}^\infty \frac{j}{(2\pi i)^2}\oint_{\infty} \frac{dz}{z^{j+1} f(z)^k}\oint_\infty \frac{w^{j-1}d w}{f(w)^m}\\
&=-\frac{1}{km}\sum_{j=1}^\infty\frac{1}{j} \frac{1}{(2\pi i)^2} \oint_{\infty} \frac{d z^{-j}}{f(z)^k}\oint_{\infty} \frac{d w^j }{f(w)^m}\\
&= \frac{1}{(2\pi i)^2}  \oint_{\infty}  \frac{d z}{z^{k+1}} \oint_{\infty} \frac{d w}{w^{m+1}} \log\left(1-\frac{h(w)}{h(z)}\right)\\
&= \frac{1}{(2\pi i)^2}  \oint_{\infty}  \frac{d z}{z^{k+1}} \oint_{\infty} \frac{d w}{w^{m+1}} \log\left(\frac{h(z)-h(w)}{z-w}\right).
\end{split}
\end{equation}
Here in the transition from the second to the third line we use the change of variables from $z$ and $w$ to $f(z)$ and $f(w)$. To derive the last line we use the positivity of  $k$ and $m$, and the function there is chosen in such way that it possess formal Taylor series expansion
\be
 \log\left(\frac{h(z)-h(w)}{z-w}\right)\in {\mathbb C}[\![z,w]\!].
\ee
This completes the proof.
\end{proof}

%%%%%%%%%%%%%%%%%%%%%%%%%%%%%%%%%%%%%

\subsection{Equivalence of Givental and Virasoro group elements}

Let us compare the subgroup  (\ref{Goper}) of the Virasoro group and Givental group, given by Lemma \ref{lem_G}. First, we identify
\be
T_k\equiv(2k+1)!!t_{2k+1}.
\ee
Second, we consider the action of two groups on the Fock space of functions of odd variables ${\bf t}^o=\{t_1,t_3,t_5,\dots\}$. On this space for the element of the Virasoro group $\widehat{V}$ one can neglect all terms with even $k$ or $m$ in (\ref{Virex}) and all terms with even values of $k+m$ in (\ref{V0}).
Therefore, on this space the action of two operators will coincide, up to linear change and translation of variables, if and only if
\be\label{const}
V_{km}^{R} =(2k+1)!! (2m+1)!! v_{{2k+1}{2m+1}}
\ee
for all $k,m\in {\mathbb Z}_{\geq 0}$. 
\begin{remark}
There is no restriction for the elements $v_{km}$ when $k$ or $m$ is even. However, we will see that the constraint (\ref{const}) is already very restrictive.
\end{remark}
\begin{remark}
Arguments of this section are closely related to the standard manipulations in  the context of TR/GD, often referred to as the {\em Laplace transform}, see, e.g., \cite{TV,Eynard2,DBOSS}. 
\end{remark}

 Let
\be\label{xz}
x(z):=\frac{f(z)^2}{2}=\frac{z^2}{2}+z^3{\mathbb C}[\![z]\!],
\ee
where $f(z)$ describes the Virasoro group element by (\ref{ffunct}) and
 \be
 N(z):=\frac{z}{x'(z)}=1+\sum_{k\in \z_{>1}}n_k z^k.
 \ee
We fix the gauge by $a_1=0$,  which immediately leads to $n_1=0$. Then we have
\begin{lemma}
From (\ref{const}) it follows that $n_k=0$ for $k>3$.
\end{lemma}
\begin{proof}
From (\ref{Vgiv}) and Lemma \ref{lem_log} we see that the relation (\ref{const}) is equivalent to
\be\label{sra}
\frac{1-R(-w)R(-z)}{w+z}=\frac{1}{2\pi (zw)^{3/2}}\int_{\gamma^2} d\eta_1  d\eta_2\,\eta_1 \eta_2  e^{-\frac{\eta_1^2}{2z}-\frac{\eta_2^2}{2w}} \log\left(\frac{h(\eta_1)-h(\eta_2)}{\eta_1-\eta_2}\right).
\ee
Here $\gamma$ is a local steepest descent contour in the vicinity of $\eta_1=\eta_2=0$, and we consider the asymptotic expansion of the integral at small values of $|z|$ and $|w|$. Let $x_k=x(\eta_k)$. The right hand side of (\ref{sra}) after the change of variables of integration $\eta_k \mapsto f(\eta_k)$ is equal to
\begin{equation}
\begin{split}\label{newrr}
&\frac{1}{2\pi (zw)^{3/2}}\int_{\gamma^2} dx_1  dx_1   \log\left(\frac{\eta_1-\eta_2}{f(\eta_1)-f(\eta_2)}\right) e^{-\frac{x_1}{z}-\frac{x_2}{w}}\\
&=-\frac{1}{2\pi (z+w)\sqrt{zw}}\int_{\gamma^2} dx_1  dx_2  \log\left(\frac{\eta_1-\eta_2}{f(\eta_1)-f(\eta_2)}\right) \left(\frac{\p}{\p x_1}+\frac{\p}{\p x_2}\right)e^{-\frac{x_1}{z}-\frac{x_2}{w}}\\
&=\frac{1}{2\pi (z+w)\sqrt{zw}}\int_{\gamma^2} dx_1  dx_2\,  e^{-\frac{x_1}{z}-\frac{x_2}{w}}   \left(\frac{\p}{\p x_1}+\frac{\p}{\p x_2}\right) \log\left(\frac{\eta_1-\eta_2}{f(\eta_1)-f(\eta_2)}\right),
\end{split}
\end{equation}
where the last equality follows form the integration by parts. From the identity
\be
\frac{\p x_1}{\p \eta_1}\frac{\p x_2}{\p \eta_2}\,\left(\frac{\p}{\p x_1}+\frac{\p}{\p x_2}\right)\log\left(\frac{\eta_1-\eta_2}{f(\eta_1)-f(\eta_2)}\right)= f'(\eta_1)f'(\eta_2)-\frac{\frac{\p x_1}{\p \eta_1} -\frac{\p x_2}{\p \eta_2}}{\eta_1-\eta_2}
\ee
it follows that the last line of (\ref{newrr}) is equal to 
\be
\frac{1}{2\pi (z+w)\sqrt{zw}}\int_{\gamma^2} \left(d\sqrt{2x_1} d\sqrt{2x_2}  - \frac{d x_1 d \eta_2-d \eta_1 d x_2}{\eta_1-\eta_2}\right) \,  e^{-\frac{x_1}{z}-\frac{x_2}{w}}.  
\ee
Comparing it with (\ref{sra}), we see that the latter is equivalent to
\be\label{RR}
R(-w)R(-z)= \frac{1}{2\pi  \sqrt{zw}}\int_{\gamma^2}  \frac{d x_1d \eta_2-d \eta_1 d x_1}{\eta_1-\eta_2} e^{-\frac{x_1}{z}-\frac{x_2}{w}}.
\ee
The left hand side factorizes, so should the right hand side. The factors can be easily found if one puts $w=0$,
\be\label{Rz}
R(-z)= \frac{1}{\sqrt{2\pi z} }\int_{\gamma}  \frac{dx (\eta) }{\eta} e^{-\frac{x(\eta)}{z}}.
\ee
Hence $x(\eta)$ should satisfy the equation
\be\label{inta}
 \frac{1}{2\pi  \sqrt{zw}}\int_{\gamma^2} \left( \frac{d x_1d \eta_2-d \eta_1 d x_2}{\eta_1-\eta_2}  -\frac{d x_1d x_2}{\eta_1\eta_2}\right)e^{-\frac{x_1}{z}-\frac{x_2}{w}}=0.
 \ee
Let us solve this equation for $x(\eta)$. First, we rewrite the left hand side as
\begin{equation}
\begin{split}\label{APR}
  &\frac{1}{2\pi  \sqrt{zw}}\int_{\gamma^2} \frac{d\eta_1  d\eta_2}{N(\eta_1)N(\eta_2)} \left(\frac{\eta_1N(\eta_2)-\eta_2N(\eta_1)}{\eta_1-\eta_2}-1\right) e^{-\frac{x_1}{z}-\frac{x_2}{w}}\\
&=- \frac{1}{2\pi  \sqrt{zw}}\int_{\gamma^2} d\eta_1  d\eta_2\, \eta_1 \eta_2 \sum_{k=2}^\infty n_k \sum_{j+m=k-2} h(\eta_1)^j h(\eta_2)^m e^{-\frac{\eta_1^2}{2z}-\frac{\eta_2^2}{2w}},
\end{split}
\end{equation}
where we apply the change of integration variables $\eta_k \mapsto h(\eta_k)$ and integrate by parts.

For any $g(\eta)\in \mathbb C[\![\eta]\!]$ let 
\be
\left[g(\eta)\right]_\eta:=\frac{1}{2}\left(g(\eta)-g(-\eta)\right)
\ee
be its odd part, and for a series of two variables $g(\eta_1,\eta_2)\in {\mathbb C}[\![\eta_1,\eta_2]\!]$
\be
\left[g(\eta_1,\eta_2)\right]_{\eta_1,\eta_2}:=\frac{1}{4}\left(g(\eta_1,\eta_2)-g(-\eta_1,\eta_2)-g(\eta_1,-\eta_2)+g(-\eta_1,-\eta_2)\right).
\ee
Then (\ref{inta}) is equivalent to
\be\label{Msu}
\left[\sum_{k=2}^\infty n_k \sum_{j+m=k-2} h(\eta_1)^j h(\eta_2)^m\right]_{\eta_1,\eta_2}=0,
\ee
so we need to prove that it implies $n_k=0$ for $k>3$. It is obvious that it implies $n_4=0$. Let
\be
M_j:=\left[\sum_{k=5}^\infty n_k h(\eta_1)^{k-j-2}\right]_{\eta_1}.
\ee

It is easy to find the first terms of the series $h(x)$:
\be
h(\eta)=\eta+\frac{n_2}{4}\eta^3+\frac{n_3}{5}\eta^4+\frac{5n_2^2}{96}\eta^5+O(\eta^6).
\ee
Then
\be
\left[\sum_{k=5}^\infty n_k \sum_{j+m=k-2} h(\eta_1)^j h(\eta_2)^m\right]_{\eta_1,\eta_2}=M_1\eta_2^1+(M_3+\frac{n_2}{4}M_1) \eta_2^3+O(\eta_2^5).
\ee
Hence, from (\ref{Msu}) it follows that
\be\label{Ms}
M_1=0,\,\,\,\,\,
M_3=0.
\ee

Let us prove by induction that these two equations have the only solution $n_k=0$ for $k>4$. Assume that $n_k=0$ for all ${k=4,5,\dots,m}$. If $m$ is odd, then from the equation $M_1=0$ it follows that $n_{m+1}=0$, so we can assume that $m$ is even, $m=2k_0$. Assume that $n_{2k_0+1}\neq 0$ for some $k_0\geq 2$. There are two possibilities:

$\bullet$ Let $n_3=0$. Then
\be
h(\eta)+h(-\eta)=2 \frac{n_{2k_0+1}}{2k_0+3}\eta^{2k_0+2}+O(\eta^{2k_0+2})
\ee
and from the equation $M_1=0$ it follows that $n_{2k}=0$ for all $k\in\{3,\dots,2k_0\}$. Then 
\begin{equation}
\begin{split}
M_1&=\left(\frac{2k_0-2}{2k_0+3}n_{2k_0+1}^2+n_{4k_0+2}\right)\eta_1^{4k_0-1}+O(\eta_1^{4k_0+1}),\\
M_3&=\left(\frac{2k_0-4}{2k_0+3}n_{2k_0+1}^2+n_{4k_0+2}\right)\eta_1^{4k_0-3}+O(\eta_1^{4k_0-1}).
\end{split}
\end{equation}
Hence from (\ref{Ms}) we have $n_{2k_0+1}=0$, which contradicts the assumption. 

$\bullet$ Let $n_3\neq0$. Then from the first equation in (\ref{Ms}) we immediately see that $n_{2k_0+2}=0$, and
\begin{equation}
\begin{split}
M_1&=\left(\frac{2k_0-2}{5}n_3n_{2k_0+1}+n_{2k_0+4}\right)\eta_1^{2k_0+1}+O(\eta_1^{2k_0+3}),\\
M_3&=\left(\frac{2k_0-4}{5}n_3n_{2k_0+1}+n_{2k_0+4}\right)\eta_1^{2k_0-1}+O(\eta_1^{2k_0+2}).
\end{split}
\end{equation}
Hence from (\ref{Ms}) we have $n_{2k_0+1}=0$, which contradicts the assumption. This completes the proof.
\end{proof}

Therefore, to identify the rank one Givental operators with the elements of the Heisenberg-Virasoro group we can restrict ourself to a two-dimensional family of the Virasoro operators (\ref{Goper}), identification of the translation part will be considered later. Let us choose a more convenient parametrization of this family. For this purpose we use the gauge transformation (\ref{gauge}). Let us denote by $\widehat{V}_{n_1,n_2,n_3}$ the element of the Virasoro subgroup (\ref{Goper}) with function $f(z)$  given by
\be
\frac{f(z)^2}{2}=\int_{0}^z \frac{\eta d \eta}{1+n_1\eta+n_2\eta^2+n_3\eta^3}.
\ee
\begin{lemma}
\be
 \widehat{V}_{n_1,n_2,0}=e^{\frac{n_1}{3}\widehat{L}_1} \widehat{V}_{0,n_2-n_1^2/3,-n_1n_2/3+2n_1^3/27}.
\ee
\end{lemma}
\begin{proof}
Because of the one-to-one correspondence between the group (\ref{Goper}) and the space of formal series $z+z^2{\mathbb C}[\![z]\!]$, the statement of the lemma follows from the identity
\be
\int_{0}^{\frac{z}{1+n_1z/3} }\frac{\eta d \eta}{1+(n_2-n_1^2/3)\eta^2+(2n_1^3/27-n_1n_2/3)\eta^3}=\int_{0}^z \frac{\eta d \eta}{1+n_1\eta+n_2\eta^2}
\ee
and the transformation rule (\ref{gauge}).
\end{proof}
Its easy to see, that for any $\tilde{n}_2$ and $\tilde{n}_3$ the equations
\be
\tilde{n}_2=n_2-n_1^2/3,\,\,\,\,\,\,\,\,\,\,\,\,\,\, \tilde{n}_3=-n_1n_2/3+2n_1^3/27
\ee
have solutions. Hence, we can work with the 2-parameter subspace of the Virasoro group, described by
\be
\frac{f(z)^2}{2}=\int_{0}^z \frac{\eta d \eta}{1+n_1\eta+n_2\eta^2}.
\ee
Let us parametrize it by $p$ and $q$
\be
1+n_1z+n_2z^2=(1+\sqrt{p+q}z)(1+q z/\sqrt{p+q}),
\ee
or, equivalently, 
\be
q=n_2,\,\,\,\,\,\ p=\frac{n_1^2-2n_2\pm \sqrt{n_1^4-4n_1^2n_2}}{2}.
\ee
We assume that $p+q\neq 0$, then $x$ is defined by
\be\label{qwer}
d x=f(z) d f(z)=\frac{z dz}{(1+\sqrt{p+q}z)(1+q z/\sqrt{p+q})}.
\ee 
In this parametrization, if $p\neq 0$ and $q\neq 0$ we have
\be\label{fpq}
x(z)=\frac{p+q}{pq}\log\left(1+\frac{qz}{\sqrt{p+q}}\right)-\frac{1}{p}\log(1+\sqrt{p+q}z).
\ee
For $q=0$ it degenerates to
\be\label{xq}
x(z)=\frac{z}{\sqrt{p}}-\frac{1}{p}\log(1+\sqrt{p}z),
\ee
and for $p=0$ it degenerates to
\be
x(z)=\frac{1}{q}\log(1+\sqrt{q}z)-\frac{1}{\sqrt{q}}\frac{z}{1+\sqrt{q} z}.
\ee

Let us denote
\be
I(z):=\frac{1}{\sqrt{2\pi z} }\int_{\gamma}  \frac{d x(\eta) }{\eta} e^{-\frac{x(\eta)}{z}}.
\ee
It appears that if the factorization constraint (\ref{RR}) is satisfied, then corresponding $R(z)$ satisfies the symplectic condition. 
Namely, let us consider a two-parametric family of functions satisfying this  condition,
\be\label{pqfam}
R_{q,p}(z)=\exp\left({-\sum_{k=1}^\infty \frac{B_{2k}}{2k(2k-1)}\left(p^{2k-1}+q^{2k-1}-\left(\frac{pq}{p+q}\right)^{2k-1}\right)z^{2k-1}}\right).
\ee
Here $B_{2k}$ are the Bernoulli numbers
\be
\frac{xe^x}{e^x-1}=1+\frac{x}{2}+\sum_{k=1}^\infty \frac{B_{2k} x^{2k}}{(2k)!}.
\ee

\begin{lemma}\label{idi}
\be
I(z)=R_{q,p}(-z).
\ee
\end{lemma}
\begin{proof}
We need to consider three different cases.

1) Let $p\neq 0$ and $q\neq 0$. Then $x(z)$ is given by (\ref{fpq}). From the integration by parts we have
\be\label{qry}
\int_{\gamma} dx(\eta)\,  e^{-\frac{x(\eta)}{z}}=0,
\ee
hence
\begin{equation}
\begin{split}
I(z)&=\frac{1}{\sqrt{2\pi z} }\int_{\gamma} dx(\eta)  \left(\frac{1}{\eta}-q \sqrt{p+q}\right) e^{-\frac{x(\eta)}{z}}\\
&= \frac{1}{\sqrt{2\pi z} }\int_{\gamma}  \frac{d\eta}{(1+\sqrt{p+q}\eta)^{1-\frac{1}{pz}}(1+q \eta/\sqrt{p+q})^{\frac{p+q}{pqz}}}.
\end{split}
\end{equation}
Let us make a change of integration variable $\eta\mapsto \frac{p}{q\sqrt{p+q}}\eta-\frac{1}{\sqrt{p+q}}$, then
\be\label{intsm}
I(z)=\frac{1}{\sqrt{2 \pi z (p+q)}}\frac{(p+q)^{\frac{1}{pz}+\frac{1}{qz}}}{q^\frac{1}{pz}p^\frac{1}{qz}}\int_{\tilde{\gamma}} \frac{\eta^{\frac{1}{pz}-1} d\eta}{(1+\eta)^{\frac{1}{pz}+\frac{1}{qz}}}.
\ee
Here $\tilde{\gamma}$ is the new contour, corresponding to the saddle point at $\eta=q/p$. We can deform the integration contour in such a way that the asymptotic expansion does not change, 
which allows us to identify it with the asymptotic expansion of the beta-function
\be
B(a,b)=\int_{0}^\infty \frac{\eta^{a-1} dx}{(1+\eta)^{a+b}}
\ee
for $a=1/pz$ and $b=1/qz$. From Stirling's expansion
\be\label{Stir}
\Gamma(z^{-1})=\sqrt{2\pi}e^{-z}z^{z-1/2} e^{\sum_{k=1}^\infty \frac{B_{2k}}{2k(2k-1)}z^{2k-1}}
\ee
one has the asymptotic expansion of the beta-function
\be
B({(qz)}^{-1},{(pz)}^{-1})=\frac{\sqrt{2\pi z(p+q)} q^\frac{1}{pz}p^\frac{1}{qz}}{(p+q)^{\frac{1}{pz}+\frac{1}{qz}}} e^{\sum_{k=1}^\infty \frac{B_{2k}}{2k(2k-1)}\left((pz)^{2k-1}+(qz)^{2k-1}-(\frac{pq}{p+q}z)^{2k-1}\right)}.
\ee
The statement of the lemma follows from the comparison of this expansion with (\ref{intsm}).

2) Let $q=0$. Then $x(z)$ is given by (\ref{xq}), and
\be\label{Int1}
I(z)=\frac{1}{\sqrt{2\pi z}}\int_{\gamma} \frac{d\eta}{(1+\sqrt{p}\eta)^{1-\frac{1}{pz}}}e^{-\frac{\eta}{z\sqrt{p}}}.
\ee
After a change of the integration variable, $\eta\mapsto z\sqrt{p}\eta-\frac{1}{\sqrt{p}}$, we get
\be
I(z)=\frac{1}{\sqrt{2\pi zp}}e^\frac{1}{zp}(zp)^\frac{1}{zp}\int_{\tilde{\gamma}} d\eta \eta^{\frac{1}{pz}-1}e^{-\eta}.
\ee
The asymptotic expansion of the integral coincides with the asymptotic expansion of the gamma-function
\be
\Gamma(z^{-1})=\int_0^\infty d\eta \eta^{\frac{1}{z}-1}e^{-\eta},
\ee
and the statement of lemma follows from the Stirling expansion (\ref{Stir}).

3) Let $p=0$. This case can be reduced to the previous one. Namely, for $p=0$ we have
\be
I(z)=\frac{1}{\sqrt{2\pi z}}\int_{\gamma} \frac{d\eta}{(1+\sqrt{q}\eta)^{2+\frac{1}{qz}}}e^{\frac{1}{\sqrt{q}z}\frac{\eta}{1+\sqrt{q}\eta}}
\ee
 which reduces to (\ref{xq}) after the change of integration variable $\eta\mapsto -\eta/(1+\sqrt{q}\eta)$ and identification of $q$ and $p$ if we apply (\ref{qry}). This completes the proof.
\end{proof}

Below we assume that coefficients ${\bf a}$ are fixed by (\ref{ffunct}), where the function $f(z)$ is given by (\ref{qwer}). Let us consider the operator
\be\label{Dop}
D_{q,p}=-e^{{\sum_{k\in\z_{>0}} a_k{\mathtt l}_k}}\,  \frac{1}{z} \frac{\p}{\p z} \,  e^{-{\sum_{k\in \z_{>0}} a_k{\mathtt l}_k}}
=-\frac{\p}{\p x},
\ee
or, equivalently,
\be
D_{q,p}=-\frac{(1+\sqrt{p+q}z)(1+q z/\sqrt{p+q})}{z} \frac{\p}{\p z}.
\ee
Then the functions 
\be
\phi_k(z):=D_{q,p}^k\cdot \frac{1}{z}
\ee
define a change of variables if we associate $k t_k$ with $z^{-k}$ and $T_k^{q,p}$ with $\phi_k(z)$. This change of variables can be described by the recursion
\be\label{changeof}
T_0^{q,p}({\bf t})=t_1,\\
T_k^{q,p}({\bf t})=\left(q\widehat{L}_{0}  +\frac{2q+p}{\sqrt{p+q}}\widehat{L}_{-1}+\left(\widehat{L}_{-2}-\frac{t_1^2}{2}\right)\right)\cdot T_{k-1}^{q,p}.
\ee
Let us prove that the linear change of variables (\ref{changeof}) gives a transformation, required for identification of the Givental and Heisenberg-Virasoro operators.

\begin{lemma}\label{Chang_T}
\be
 T^R_k({\bf T}_{q,p}({\bf t}))= \widehat{V}_0\cdot (2k+1)!! t_{2k+1}, \,\,\,\,\,\, k\in {\mathbb Z_{\geq 0}}.
 \ee
\end{lemma}
\begin{proof}
Let us prove an equivalent relation
\be
 \widehat{V}_0^{-1}\cdot \sum_{k=0}^\infty T^R_k({\bf T}_{q,p}({\bf t})) w^k = \sum_{k=0}^\infty (2k+1)!! t_{2k+1} w^k.
\ee
Let us associate $k t_k$ with $z^{-k}$. Then $ \sum_{k=0}^\infty T^R_k w^k $ is associated with
\be
R(-w)\sum_{k=0}^\infty w^k \left(-\frac{1}{z}\frac{\p}{\p z}\right)^k \cdot \frac{1}{z}, 
\ee
and, by  (\ref{Dop}),  the sum $\sum_{m=0}^\infty T^R_m({\bf T}_{q,p}({\bf t})) w^m$ is associated with
\be\label{gty}
R(-w)\sum_{m=0}^\infty w^m D_{q,p}^m \cdot \frac{1}{z}.
\ee

Let us omit the factor $R(-w)$ for a moment and act on (\ref{gty})  by $e^{-{\sum_{\z_{>0}} a_k{\mathtt l}_k}}$: 
\begin{equation}
\begin{split}\label{intt}
\left(e^{-{\sum_{\z_{>0}} a_k{\mathtt l}_k}}\cdot \sum_{m=0}^\infty w^m D_{q,p}^m \cdot \frac{1}{z}\right)_-&=\left( \sum_{m=0}^\infty w^m \left(-\frac{1}{z}\frac{\p}{\p z}\right)^m e^{-{\sum_{\z_{>0}} a_k{\mathtt l}_k}}\cdot \frac{1}{z}\right)_-\\
&=\left( \sum_{m=0}^\infty w^m \left(-\frac{1}{z}\frac{\p}{\p z}\right)^m \frac{1}{h(z)}\right)_-\\
&= \sum_{k,m=0}^\infty \frac{1}{2 \pi i} \frac{w^m}{z^{2k+1}}\oint \frac{d\eta}{h(\eta)}  \left(\frac{\p}{\p \eta}\frac{1}{\eta}\right)^m \eta^{2k}\\
&= \sum_{k,m=0}^\infty \frac{(2k-1)!!}{2\pi i} \frac{w^{k+m}}{z^{2k+1}}  \oint \frac{d\eta}{h(\eta)} \frac{(-1)^m (2m-1)!!}{\eta^{2m}}\\
&= \sum_{k=0}^\infty (2k-1)!! \frac{w^{k}}{z^{2k+1}}  \frac{1}{\sqrt{-2 \pi z}}  \int_{\gamma} dx \frac{x}{h(x)} e^{\frac{x^2}{2z}}\\
&=R(w) \sum_{k=0}^\infty(2k-1)!! \frac{w^{k}}{z^{2k+1}},
\end{split}
\end{equation}
where  for any series $\left(\sum_{k\in\z} b_k z^k\right)_-:=\sum_{k\in\z_{<0}} b_k z^k$. Hence, 
\be
\left(e^{-{\sum_{\z_{>0}} a_k{\mathtt l}_k}}\cdot R(-w)\sum_{m=0}^\infty w^m D_{q,p}^m \cdot \frac{1}{z}\right)_-=\sum_{k=0}^\infty(2k-1)!! \frac{w^{k}}{z^{2k+1}},
\ee
and two linear transformations coincide.  
\end{proof}

Let us also introduce
\be\label{vdef}
v_k:=[z^k]\int_{0}^z (f(\eta)-y(\eta)) d x(\eta), 
\ee
where 
\be\label{ygr}
y(z) =\int^z \frac{d x(\eta)}{\eta}.
\ee
For $p\neq 0$ and $q\neq 0$ 
\be
y(z) =\frac {\sqrt {p+q}}{p}   \left( \log  \left( 1+\sqrt {p+q}z \right)  -\log  \left(1+\frac{qz}{ \sqrt {p+q}} \right) \right).
\ee
For $q=0$ it reduces to
\be
y(z)=\frac {1}{\sqrt{p}} \log  \left( 1+\sqrt {p }z \right)
\ee
and for $p=0$ it degenerates to
\be
y(z)=\frac{z}{1+\sqrt{q}z}.
\ee
It is easy to see that $v_k=0$ for $k<4$.

Having in mind that the coefficients $v_k$ are given by (\ref{vdef}), let us remind the reader the main ingredients of the following theorem, given by (\ref{Vird}), (\ref{changeof}), (\ref{ffunct}), (\ref{qwer}), and (\ref{pqfam}), namely the Virasoro operators
\be
\widehat{L}_m= \sum_{k=1}^\infty k t_k \frac{\p}{\p t_{k+m}}+\frac{1}{2} \sum_{a+b=m} \frac{\p^2}{\p t_a \p t_b},\,\,\,\,\,\,\,\,\, m\geq0,
\ee
the linear change of variables
\be
T_0^{q,p}({\bf t})=t_1,\\
T_k^{q,p}({\bf t})=\left(q\widehat{L}_{0}  +\frac{2q+p}{\sqrt{p+q}}\widehat{L}_{-1}+\left(\widehat{L}_{-2}-\frac{t_1^2}{2}\right)\right)\cdot T_{k-1}^{q,p},
\ee
and the coefficients $a_k$ are given implicitly through the formal series 
\be
f(z):=e^{{\sum_{k\in\z_{>0}} a_k{\mathtt l}_k}}\,  z \,  e^{-{\sum_{k\in\z_{>0}} a_k{\mathtt l}_k}} \in z+z{\mathbb C}[\![z]\!],
\ee
which is defined by
\be
f(z) d f(z)=\frac{z dz}{(1+\sqrt{p+q}z)(1+q z/\sqrt{p+q})},
\ee 
and a formal series
\be\label{pqfam1}
R_{q,p}(z)=\exp\left({-\sum_{k=1}^\infty \frac{B_{2k}}{2k(2k-1)}\left(p^{2k-1}+q^{2k-1}-\left(\frac{pq}{p+q}\right)^{2k-1}\right)z^{2k-1}}\right),
\ee
which defines an element of the Givental group. Then
\begin{theorem}\label{th_RL}
Rank one Givental operator coincides up to a  linear change of variables with the element of the Heisenberg-Virasoro group of symmetries of KP hierarchy if and only if $R(z)$ belongs to a family (\ref{pqfam1}). For this family
\be\label{RL}
\left.\widehat{R}_{q,p}\cdot Z({\bf t}^{o})\right|_{{\bf T}={\bf T}^{q,p}({\bf t})}=  e^{\hbar^{-1} \sum_{k=4} v_k \frac{\p}{\p t_{k}}}e^{\sum_{k\in\z_{>0}} a_k \widehat{L}_k} \cdot Z({\bf t}^{o})
\ee
for any function $Z({\bf t}^{o})$.
\end{theorem}
\begin{proof}
We have already identified the quadratic parts of two operators and linear changes of variables in Lemmas \ref{idi} and \ref{Chang_T}.
It remains to identify the translations of the variables. Let us show that for any $Z({\bf t}^{o})$
\be\label{psiq} 
\widehat{V}_0^{-1}\left( \sum_{k=4}^\infty v_k \frac{\p}{\p t_{k}} \right) \widehat{V}_0 \cdot Z({\bf t}^{o})=  \sum_{k=2}^\infty \delta_k \frac{\p}{\p T_k}Z({\bf t}^{o}),
\ee
 where $\delta_k$'s are given by (\ref{shiftsR}). By definition,
\begin{equation}
\begin{split}
e^{-{\sum_{\z_{>0}} a_k{\mathtt l}_k}}\, \sum_{k=4}^\infty v_k z^k \,e^{{\sum_{\z_{>0}} a_k{\mathtt l}_k}} &=\int_{0}^{h(z)} (f(\eta)-y(\eta)) d x(\eta)\\
&=\int_{0}^{z} (\eta-y(h(\eta))) \eta d \eta.
\end{split}
\end{equation}
Using integration by parts we have
\begin{equation}
\begin{split}
\frac{1}{\sqrt{2\pi z}}\int_{\gamma} d \eta (\eta-y(h(\eta)))\eta e^{-\frac{\eta^2}{2z}}&= z - \frac{z}{\sqrt{2\pi z}}\int_{\gamma} d y(\eta)e^{-\frac{x(\eta)}{z}}\\
&=z(1-R(-z)).
\end{split}
\end{equation}
This completes the proof.
\end{proof}

Let us stress that this theorem essentially describes the equality of two operators, therefore the formula is true for any function $Z({\bf t}^{o})$ for which both sides make sense, not only for the solutions of the KdV hierarchy.

\begin{remark}
The identification of the Givental and Heisenberg-Virasoro operators (\ref{RL}) for $p=0$ was established in \cite{LiuWang}.
\end{remark}

It is natural to compare the higher rank Givental group with the symmetry group of the multicomponent KP hierarchy. However, for higher rank Givental groups there are other promising possibilities,  and one can try to identify the Givental operators with the operators from the symmetry groups of other solitonic integrable hierarchies. In particular, for rank two, one can consider 2D Toda lattice and extended Toda lattice. It would be interesting to find a relation between our analysis and the vertex operator analysis of the bilinear Hirota equations considered in \cite{FGM,GM} for the ADE type singularities.

\begin{remark}
For any semisimple cohomological field theory one can derive the Virasoro constraints by conjugation of the constraints for the KdV tau-functions with the Givental group element. However, for the cases, when after a linear change of variables the generating function coincides with a tau-function of integrable hierarchy, these constraints gain additional properties. Namely, after a linear change of variables, they belong to the  symmetry algebra of the hierarchy. Virasoro type constrains and their conjugation by the Givental operators can provide alternative way for analysis of the integrability and other interesting  properties.
\end{remark}

%%%%%%%%%%%%%%%%%%%%%%%%%%%%%%%%%

\section{Hodge integrals}\label{S_Hodge}

\subsection{Triple Hodge integrals}

Denote by~$\mathcal{M}_{g,n}$ the moduli space of all compact Riemann surfaces of genus~$g$ with~$n$ distinct marked points. Deligne and Mumford \cite{DM69} defined a natural compactification $\mathcal{M}_{g,n}\subset\overline{\mathcal{M}}_{g,n}$ via stable curves with possible nodal singularities. The moduli space $\overline{\mathcal{M}}_{g,n}$ is a non-singular complex orbifold of dimension~$3g-3+n$. It is empty unless the stability condition
\begin{gather}\label{closed stability}
2g-2+n>0
\end{gather}
is satisfied. We refer the reader to~\cite{DM69,HM98} for the basic theory.

In his seminal paper~\cite{Wit91}, Witten initiated new directions in the study of $\overline{\mathcal{M}}_{g,n}$. For each marking index~$i$ consider the cotangent line bundle ${\mathbb{L}}_i \rightarrow \overline{\mathcal{M}}_{g,n}$, whose fiber over a point $[\Sigma,z_1,\ldots,z_n]\in \overline{\mathcal{M}}_{g,n}$ is the complex cotangent space $T_{z_i}^*\Sigma$ of $\Sigma$ at $z_i$. Let $\psi_i\in H^2(\overline{\mathcal{M}}_{g,n},\mathbb{Q})$ denote the first Chern class of ${\mathbb{L}}_i$. We consider the intersection numbers
\begin{gather}\label{eq:products}
\<\tau_{a_1} \tau_{a_2} \cdots \tau_{a_n}\>_g:=\int_{\overline{\mathcal{M}}_{g,n}} \psi_1^{a_1} \psi_2^{a_2} \cdots \psi_n^{a_n}.
\end{gather}
The integral on the right-hand side of~\eqref{eq:products} vanishes unless the stability condition~\eqref{closed stability} is satisfied, all  $a_i$ are non-negative integers, and the dimension constraint $3g-3+n=\sum a_i$ holds true. Let $T_i$, $i\geq 0$, be formal variables and let
\be
\tau_{KW}:=\exp\left(\sum_{g=0}^\infty \sum_{n=0}^\infty \hbar^{2g-2+n}F_{g,n}\right),
\ee
where
\be
F_{g,n}:=\sum_{a_1,\ldots,a_n\ge 0}\<\tau_{a_1}\tau_{a_2}\cdots\tau_{a_n}\>_g\frac{\prod T_{a_i}}{n!}.
\ee
Witten's conjecture \cite{Wit91}, proved by Kontsevich \cite{Kon92}, states that the partition function~$\tau_{KW}$ becomes a tau-function of the KdV hierarchy after the change of variables~$T_n=(2n+1)!!t_{2n+1}$. Integrability immediately follows \cite{KMMMZ92} from Kontsevich's matrix integral representation, for more details see the companion paper \cite{H3_2}.

There are different interesting deformations of this tau-function. Let us consider the Hodge bundle $\mathbb E$, a rank $g$ vector bundle over ${\mathcal{M}}_{g,n}$. Let $\ch_k\in H^{2k}(\overline{\mathcal{M}}_{g,n})$ be the components of the Chern character of $\mathbb E$. According to Mumford, the even components of the Chern character of $\mathbb E$ identically vanish. Then we consider
\be\label{freecharac}
\tilde{\mathcal{F}}_{g,n}=\sum_{a_1,\ldots,a_n\ge 0}  \frac{\prod T_{a_i}}{n!} \int_{\overline{\mathcal{M}}_{g,n}} e^{\sum_{k=1}^\infty s_k \ch_{2k-1}} \psi_1^{a_1} \psi_2^{a_2} \cdots \psi_n^{a_n}.
\ee
We introduce a generating function
\be \label{generalgener}
Z({\bf T}; {\bf s})=\exp\left(\sum_{g=0}^\infty \sum_{n=0}^\infty \hbar^{2g-2+n}\tilde{\mathcal{F}}_{g,n}\right).
\ee
In general, this deformation does not preserve integrability, because it is easy to check that $Z({\bf T}; {\bf s})$ is not a tau-function of the KdV hierarchy in the variables $t_k$ anymore. Investigation of the integrable properties of $Z({\bf t}; {\bf s})$ is one of the main goals of this paper.

From Mumford's theorem \cite{Mum}  it follows that
\be\label{MuM}
Z({\bf T}; {\bf s})=\widehat{R}({\bf s})\cdot \tau_{KW}({\bf T}),
\ee
where 
\be
\widehat{R}({\bf s}):= \exp\left( {\sum_{k=1}^\infty \frac{B_{2k}}{(2k)!} s_{k} \widehat{W}_{k}}\right)
\ee
and operators $\widehat{W}_{k}$ are given by (\ref{Woper}).

For the Hodge bundle over $\overline{\mathcal{M}}_{g,n}$ let $\lambda_i=c_i({\mathbb E})$. Then
\be
\Lambda_g(u)=\sum_{i=0}^g u^i \lambda_i =e^{\sum_{m=1}^\infty(2m-2)! \ch_{2m-1} u^{2m-1}}
\ee
satisfies $\Lambda_g(u)\Lambda_g(-u)=1$. Particular specifications of the parameters ${\bf s}$ correspond to various families of Hodge integrals. The generating function of the $m$-linear Hodge integrals 
\be
\int_{\overline{\mathcal{M}}_{g,n}} \Lambda_g(u_1) \Lambda_g(u_2)\dots \Lambda_g(u_m) \psi_1^{a_1} \psi_2^{a_2} \cdots \psi_n^{a_n}
\ee
corresponds to the Miwa parametrization of ${\bf s}$ variables
\be\label{Miws}
s_k=(2k-2)!\sum_{j=1}^m u_j^{2k-1}.
\ee
%or
%\be
%\Lambda_g(u_1) \Lambda_g(u_2)\dots \Lambda_g(u_m)=e^{\sum_{k=1}^\infty (2k-2)!\sum_{j=1}^m u_j^{2k-1} \ch_{2k-1}}.
%\ee
In particular, the generating function of linear Hodge integrals corresponds to the case $s_k=(2k-2)!u^{2k-1}$. Form Givental's theory it follows that a product of the 
Hodge classes $ \Lambda_g(u)$ defines a CohFT with flat unit.
\begin{remark}
We see that it is natural to consider a Miwa-like parametrization of the ${\bf s}$ variables. It might be an indication of the existence of the integrable structure in these variables, at least for certain values of ${\bf T}$ variables. 
\end{remark}

Let us consider the case of cubic Hodge integrals with an additional Calabi-Yau condition
\be\label{specialcond}
\frac{1}{u_1}+\frac{1}{u_2}+\frac{1}{u_3}=0.
\ee
It is convenient to use the parametrization
\be\label{speccond}
u_1=-p, u_2=-q, u_3=\frac{pq}{p+q}.
\ee
\begin{remark}
Cubic Hodge integrals satisfying the Calabi-Yau condition are related to the {\em topological vertex}, which plays an important role in the topological string models on 3-dimensional Calabi-Yau manifolds. In this context the parameters $p$ and $q$ are related to the local framings, for more detail see, e.g., \cite{TV} and references therein. 
\end{remark}

Consider the generating function
\be
Z_{q,p}({\bf T})=e^{\sum_{g=0}^\infty \sum_{n=0}^\infty \hbar^{2g-2+n} {\mathcal F}_{g,n}},
\ee
where
\be
{\mathcal F}_{g,n}=\sum_{a_1,\dots,a_n}\frac{\prod T_{a_i}}{n!}\int_{\overline{\mathcal{M}}_{g,n}} \Lambda_g (-q) \Lambda_g (-p) \Lambda_g (\frac{pq}{p+q})\psi_1^{a_1} \psi_2^{a_2} \cdots \psi_n^{a_n}.
\ee

Linear Hodge integrals appear as a specification at $p=u^2$, $q=0$ (or $q=u^2$, $p=0$), because $\Lambda_g(0)=1$. From Theorem \ref{th_RL} and equation (\ref{MuM}) it immediately follows that
\begin{theorem}\label{Th_H}
Generating function of triple Hodge integrals, satisfying the Calabi-Yau condition, in the variables (\ref{changeof}) is a tau-function of the KP hierarchy,
\be
\tau_{q,p}({\bf t})=Z_{q,p}({{\bf T}^{q,p}(\bf t)}).
\ee
It is related to the Kontsevich-Witten tau-function by an element of the Heisenberg-Virasoro group 
\be\label{qptau}
\tau_{q,p}({\bf t})=e^{\hbar^{-1} \sum_{k=4} v_k \frac{\p}{\p t_{k}}}e^{\sum_{k\in\z_{>0}} a_k \widehat{L}_k} \cdot  \tau_{KW}({\bf t}).
\ee
\end{theorem}
For $p=0$ the KP integrability  was proved by Kazarian \cite{Kaza}.
\begin{remark}
After announcement of the main results of this project, Kramer \cite{Kramer} found an independent proof of KP integrability of triple Hodge integrals, satisfying the Calabi-Yau condition. This prove uses the Mari\~{n}o-Vafa formula (see \cite{Zhou} for more details).

One can reverse the logic and prove the Mari\~{n}o-Vafa formula using our proof of Theorem \ref{Th_H}. Indeed, the Mari\~{n}o-Vafa formula describes the coefficents of expansion of two KP tau-functions, related by an element of the Heisenberg-Virasoro subgroup of the KP hierarchy symmetry group. Therefore, to prove the Mari\~{n}o-Vafa formula
it is enough to prove the relation between two tau-functions of the KP hierarchy, which can be done using the relations between corresponding points of the Sato Grassmannian and Kac-Schwarz algebras. The details will be given elsewhere. 
\end{remark}

\begin{conjecture}
The case of triple Hodge integrals, satisfying the Calabi-Yau condition, is the most general case of generating functions of Hodge integrals (\ref{generalgener}), satisfying the KP hierarchy in the variables ${\bf T}$ after a linear change of variables. A direct analog of this statement is also true for the $\Theta$-Hodge integrals, see next section. 
\end{conjecture}

\begin{conjecture}\label{conjr}
Theorem \ref{Th_H} has a direct generalization for the $r$-spin case. 
\end{conjecture}
In a companion paper \cite{H3_2} we develop the methods of the KP hierarchy, suitable for the investigation of the family of the tau-functions, conjecturally describing a generalization of $\tau_{q,p}$ for the $r$-spin case.

\begin{remark}
It would also be interesting to investigate KP type integrability of Hodge integrals in the context of Gromov-Witten theory \cite{CZ,Faber}. Particularly attracting is the Gromov-Witten theory of ${\mathbb P}^1$, where the generating function is given by the tau-function of the extended Toda hierarchy. 
\end{remark}

Hodge integrals, in particular, the representatives of the Calabi-Yau family for the specific values of the parameters $q$ and $p$, are related to several other interesting integrable structures, including the Dubrovin-Zhang, Volterra or Hodge integrable hierarchies \cite{Buryak,Takasaki,Dubrovin16,Dub,Liu,Zhou}. Connection between them and the KP integrability, investigated in this paper, will be discussed elsewhere. 

It is also interesting to consider the specific values of the parameters $q$ and $p$, for which the tau-function $\tau_{q,p}$ describe $k$-reductions of the KP hierarchy.
\begin{corollary}\label{red}
For $p=-2q$ the tau-function $\tau_{q,p}$ does not depend on even times, 
\be
\frac{\p}{\p t_{2k}}\tau_{q,-2q}({\bf t})=0, \,\,\,\,\,\,\, k\in{\mathbb Z}_{>0},
\ee
that is $\tau_{q,-2q}({\bf t})$ is a tau-function of the KdV hierarchy.
\end{corollary}
\begin{proof}
It follows from the change of variables (\ref{changeof}). Indeed, since $2q+p=0$ the operator, in (\ref{changeof}) does not contain $\widehat{L}_{-1}$, and the operators $\widehat{L}_0-\frac{t_1^2}{2}$ and $\widehat{L}_{-2}$ preserve the parity of the indices of time variables. 
\end{proof}
It is easy to see that it is impossible to get higher $k$-reductions of the KP hierarchy (aka Gelfand--Dickey hierarchies) this way.

In \cite{Dubrovin16} the relation between the generating function $Z_{q,-2q}$ and the discrete KdV was established. Our results show that there is a simpler relation between $Z_{q,-2q}$ and the ordinary KdV hierarchy, given by a linear change of variables. We expect that this relation can help us better understand the relation between the  discrete  and  ordinary KdV hierarchies.
%%%%%%%%%%%%%%%%%%%%%%%%%%%%%%%%%%%%%%%%

\subsection{Triple $\Theta$-Hodge integrals}\label{S_theta}

From Theorem \ref{th_RL} it follows that the action of Givental operator $\widehat{R}_{q,p}$ on {\em any} KdV tau-function leads to the solution of the KP hierarchy after a linear change of variables. The most natural alternative to the KW tau-function here is the Br\'ezin-Gross-Witten (BGW) tau-function. This KdV tau-function governs the intersection theory with the insertions of the fascinating Norbury's $\Theta$-classes. We refer the reader to \cite{Norb,CN,NP1} for a detailed description. The role of the BGW tau-function as the universal building block of TR/GD was first observed by the author, Mironov and Morozov \cite{AMM5,AMM6,AMM8} in the context of matrix models.

Norbury's $\Theta$-classes are the cohomology classes, $\Theta_{g,n}\in H^{4g-4+2n}(\overline{\mathcal{M}}_{g,n})$, described in \cite{Norb}. Consider the generating function of the 
intersection numbers of $\Theta$-classes and $\psi$-classes
\be
F^\Theta= \sum_{a_1,\ldots,a_n\ge 0}  \frac{\prod T_{a_i}}{n!} \int_{\overline{\mathcal{M}}_{g,n}}\Theta_{g,n} \psi_1^{a_1} \psi_2^{a_2} \cdots \psi_n^{a_n} 
\ee
then we have a direct analog of the Kontsevich-Witten tau-function \cite{Norb}:
\begin{theorem*}[Norbury]
Generating function
\be
\tau_\Theta = \exp\left(\sum_{g=0}^\infty \sum_{n=0}^\infty \hbar^{2g-2+n}F_{g,n}^\Theta \right)
\ee
becomes a tau-function of the KdV hierarchy after the change of variables~$T_n=(2n+1)!!t_{2n+1}$.
\end{theorem*}

Norbury also proved, that $\tau_\Theta$ is nothing but a tau-function of the BGW model, described by a unitary matrix integral \cite{GW,Brezin}
\be
\tau_\Theta=\tau_{BGW}.
\ee
KdV integrability of the BGW model follows from the relation to the generalized Kontsevich model and was established by Mironov, Morozov and Semenoff \cite{MMS}.

Insertion of Norbury's $\Theta$-classes can be naturally absorbed by the Givental formalism \cite{NP1}. This allows us to relate integrable properties of generating functions of intersection numbers with and without $\Theta$-classes. 
Consider the symplectic transformation $R(z)$, introduced in Section \ref{Givsec}. For it we introduce the change of the dilaton shift associated to the BGW tau-function
\be\label{shiftsR0}
\sum_{k=1}^\infty \delta_k^0 z^k:=1-R(-z).
\ee
Note, that the right hand side of this equation is equal to the right hand side of (\ref{shiftsR}), divided by $z$.
For any symplectic transformation $R(z)$ consider the corresponding action of the Givental operator with new translation part
\be\label{GivN}
\widehat{R}^0\cdot Z({\bf T})=\,e^{\frac{1}{2}\sum_{i,j=0}^\infty{V_{ij}^R}\frac{\p^2}{\p T_i \p T_j}}\,\left.e^{\hbar^{-1} \sum_{k=1}^\infty \delta_k^0 \frac{\p}{\p T_k}}\,Z({\bf T})\right|_{{T}_k\mapsto {T}^R_k}.
\ee

The classes $\Theta_{g,n}$ define a degenerate CohFT. It was shown by Norbury that the multiplication of the CohFT by these classes corresponds to a simple modification of the Givental operators, given by (\ref{GivN}). It is described by Proposition 3.9 of \cite{NP1}, let us formulate a version of this proposition for the rank 1 case. We stress that the original proposition of Norbury describes a relation between CohFT's, and we consider only its corollary, a relation between the partition functions of these CohFT's. For more details see \cite{NP1}.

The action of the Givental group on the trivial CohFT, associated with the Kontsevich--Witten tau-function, produces certain classes of cohomology $\Omega_{g,n}\in H^*(\overline{\mathcal{M}}_{g,n})$. Let us consider it on the level of partition functions. For any element of the Givental group $\widehat{R}$ one has
\be
\widetilde{Z}^{\Omega}({\bf T})=\widehat{R}\cdot \tau_{KW},
\ee
where
\be
\widetilde{Z}^{\Omega}({\bf T})=\exp\left(\sum_{g=0}^\infty \sum_{n=0}^\infty \hbar^{2g-2+n}{\mathcal F}_{g,n}^\Omega \right),
\ee
with
\be
{\mathcal F}_{g,n}^{\Omega}=\sum_{a_1,\dots,a_n}\frac{\prod T_{a_i}}{n!}\int_{\overline{\mathcal{M}}_{g,n}} \Omega_{g,n} \psi_1^{a_1} \psi_2^{a_2} \cdots \psi_l^{a_n}.
\ee
For Givental's theory it follows that the classes $\Omega_{g,n}$ exist and satisfy certain peculiar properties. In particular, for the rank one case which we consider in this paper, they always can be expressed in terms of the components of the Chern character of $\mathbb E$, see (\ref{freecharac}).

Then the insertion of the $\Theta$-classes is described by
\begin{proposition}[Norbury]\label{NCohFT}
\be
\widetilde{Z}^{\Theta\Omega}({\bf T})=\widehat{R}^0\cdot \tau_{\Theta},
\ee
where $\Theta\Omega$ denotes the classes $\Theta_{g,n}\Omega_{g,n}$.
\end{proposition}

Let us apply this proposition for triple Hodge integrals, that is, $\Omega_{g,n}= \Lambda_g (-q) \Lambda_g (-p) \Lambda_g (\frac{pq}{p+q})$. Let
\be\label{GenfN}
{Z}_{q,p}^{\Theta}({\bf T})=\exp\left(\sum_{g=0}^\infty \sum_{n=0}^\infty \hbar^{2g-2+n}{\mathcal F}_{g,n}^\Theta \right),
\ee
where
\be
{\mathcal F}_{g,n}^{\Theta}=\sum_{a_1,\dots,a_n}\frac{\prod T_{a_i}}{n!}\int_{\overline{\mathcal{M}}_{g,n}} \Theta_{g,n} \Lambda_g (-q) \Lambda_g (-p) \Lambda_g (\frac{pq}{p+q})\psi_1^{a_1} \psi_2^{a_2} \cdots \psi_l^{a_n}.
\ee
From Proposition \ref{NCohFT} it immediately follows that
\begin{lemma}
\be
{Z}_{q,p}^{\Theta}({\bf T})= \widehat{R}^0_{q,p}\cdot  \tau_\Theta.
\ee
\end{lemma}

Let us introduce the coefficients
\be
 \tilde{v}^0_k=[z^k]\int_0^z \left(df(\eta)- dy(\eta)\right).
\ee
They are related to the translations (\ref{shiftsR0}),
\be
\widehat{V}_0^{-1}\left( \sum_{k=2}^\infty \tilde{v}^0_k \frac{\p}{\p t_{k}} \right) \widehat{V}_0 \cdot Z({\bf t}^{o}) =  \sum_{k=1}^\infty \delta_k^0 \frac{\p}{\p T_k}\cdot Z({\bf t}^{o}).
\ee
Indeed
\begin{equation}
\begin{split}
e^{-{\sum_{\z_{>0}} a_k{\mathtt l}_k}} \sum_{k=2}^\infty v_k^0 z^k e^{{\sum_{\z_{>0}} a_k{\mathtt l}_k}} &=\int_{0}^{h(z)}  \left(df(\eta)- dy(\eta)\right)\\
&=\int_{0}^{z} (d \eta-d y(h(\eta))), 
\end{split}
\end{equation}
and from integration by parts
\begin{equation}
\begin{split}
\frac{1}{\sqrt{2\pi z}}\int_{\gamma}  (d \eta- dy(h(\eta))) e^{-\frac{\eta^2}{2z}}&= 1 - \frac{1}{\sqrt{2\pi z}}\int_{\gamma} d y(\eta)e^{-\frac{x(\eta)}{z}}\\
&=1-R(-z).
\end{split}
\end{equation}

Then we have
\begin{theorem}\label{Norbt}
Generating function of triple $\Theta$-Hodge integrals, satisfying the Calabi-Yau condition, in the variables (\ref{changeof}) is a tau-function of the KP hierarchy,
\be
{\tau}^{\Theta}_{q,p}({\bf t})={Z}_{q,p}^\Theta({{\bf T}^{q,p}(\bf t)}).
\ee
It is related to the Br\'ezin-Gross-Witten tau-function by an element of the Heisenberg-Virasoro group 
\be\label{qptauT}
{\tau}^{\Theta}_{q,p}({\bf t})=e^{\hbar^{-1} \sum_{k>0} \tilde{v}_k^0 \frac{\p}{\p t_{k}}}e^{\sum_{k\in\z_{>0}} a_k \widehat{L}_k} \cdot  \tau_{BGW}({\bf t}).
\ee
\end{theorem}
\begin{conjecture}
Theorem \ref{Norbt} has a direct generalization for the $r$-spin case. 
\end{conjecture}

Similar to Corollary \ref{red} one has
\begin{corollary}
For $p=-2q$ the tau-function $\tau_{q,p}$ does not depend on even times, 
\be
\frac{\p}{\p t_{2k}}\tau^{\Theta}_{q,-2q}({\bf t})=0, \,\,\,\,\,\,\, k\in{\mathbb Z}_{>0},
\ee
that is $\tau^{\Theta}_{q,-2q}({\bf t})$ is a tau-function of the KdV hierarchy.
\end{corollary}
The proof is completely analogous to the proof Corollary \ref{red}.

Generating functions ${Z}_{q,p}$, ${Z}_{q,p}^{\Theta}$ satisfy the Virasoro constraints, and tau-functions ${\tau}_{q,p}$, ${\tau}^{\Theta}_{q,p}$ satisfy the Heisenberg-Virasoro constraints. All these constraints can be obtained from the constraints for $\tau_{KW}$ and $\tau_{BGW}$ by conjugation with Givental (for $Z$'s) and Heisenberg-Virasoro (for $\tau$'s) group elements. Moreover, all these generating functions possess simple cut-and-join description. These topics, as well as matrix integral description of ${\tau}^{\Theta}_{q,p}({\bf t})$ and its generalization for the $r$-spin case, will be described elsewhere.

\section*{Acknowledgments}
The author is grateful to G. Carlet,  R. Kramer, and S. Shadrin for useful discussions and the anonymous referees for the suggested improvements.
This work was supported by  IBS-R003-D1 and by RFBR grant 18-01-00926. The author would also like to thank Vasily Pestun for his hospitality at IHES supported by the European Research Council under the European Union's Horizon 2020 research and innovation programme, QUASIFT grant agreement 677368.

\end{document}